\documentclass[10pts,twocolumns]{autart}

\usepackage{url}
\usepackage{here}
\usepackage{graphicx}          
\usepackage{slashbox}                               

\usepackage{amssymb}
\usepackage{amsmath}
\usepackage{graphicx}
\usepackage{amsfonts}
\usepackage{multirow}

\usepackage{psfrag}
\usepackage{mathrsfs}
\usepackage{dsfont}

\DeclareMathOperator{\diag}{diag}

\everymath{\displaystyle}

\newtheorem{theorem}{Theorem}[section]
\newtheorem{corollary}[theorem]{Corollary}
\newtheorem{proposition}[theorem]{Proposition}
\newtheorem{lemma}[theorem]{Lemma}
\newtheorem{define}[theorem]{Definition}

\newenvironment{proof}{{\it Proof :~}}{\hfill$\diamondsuit$\\}
\newtheorem{remark}[theorem]{Remark}

\newcommand{\menddef}{\hfill \ensuremath{\triangledown}}

\DeclareMathOperator*{\T}{\intercal}
\DeclareMathOperator*{\He}{He}

\DeclareMathOperator{\eps}{\varepsilon}

\renewcommand{\theenumi}{\alph{enumi}}
\begin{document}
\begin{frontmatter}

\title{On the necessity of looped-functionals arising in the analysis of pseudo-periodic, sampled-data and hybrid systems}

\author[First]{Corentin Briat}\ead{briatc@bsse.ethz.ch,corentin@briat.info}\ead[url]{http://www.briat.info}
\author[Second]{ and Alexandre Seuret}\ead{Alexandre.Seuret@gipsa-lab.grenoble-inp.fr}\ead[url]{http://www.gipsa-lab.fr/\textasciitilde alexandre.seuret}

\address[First]{Swiss Federal Institute of Technology--Z\"{u}rich, Department of Biosystems Science and Engineering, Switzerland.}
\address[Second]{CNRS, LAAS, 7 avenue du colonel oRche, F-31400 Toulouse, France and Univ de Toulouse, LAAS, F-31400 Toulouse, France.}

\begin{abstract}
    Looped-functionals have been shown to be relevant for the analysis of a wide variety of systems. However, the conditions obtained in \cite{Seuret:13,Briat:12h,Briat:13b} for the analysis of sampled-data, impulsive and switched systems have only been shown to be sufficient for the characterization of their associated discrete-time stability conditions. We prove here that these conditions are also \emph{necessary}. This result is derived for a wider class of linear systems, referred to as impulsive pseudo-periodic systems, that encompass periodic, impulsive, sampled-data and switched systems as special cases. 
\end{abstract}

\begin{keyword}
Looped-functionals, periodic systems, sampled-data systems, hybrid systems,
\end{keyword}
\end{frontmatter}

\section{Introduction}

Looped-functionals have been introduced in \cite{Seuret:13} for the analysis of sampled-data systems. They have been then further considered in \cite{Briat:11l,Briat:12d} for the analysis of impulsive systems (see e.g. \cite{Bainov:89,Hespanha:08}) and in \cite{Briat:13b} in the context of switched systems (see e.g. \cite{Morse:96,Hespanha:99,Liberzon:03}). The main rationale behind such functionals lies in the reformulation of a discrete-time condition into an expression that allows for  the consideration of linear uncertain time-varying systems \cite{Briat:11l,Briat:12d,Briat:13b} and nonlinear systems \cite{Peet:14}.  The main difference with Lyapunov functionals, see e.g. \cite{Naghshtabrizi:08}, is that positivity is not a required condition anymore. Instead, we demand that a certain algebraic boundary condition be satisfied, the so-called \emph{looping-condition}. It has been demonstrated in the aforementioned papers that this class of functionals leads to less conservative conditions than those obtained using usual Lyapunov functionals, demonstrating then the relevance of the approach. 

In the papers \cite{Briat:12h,Briat:13b}, the considered looped-functional was shown to yield sufficient conditions for the feasibility of the discrete-time stability criterion characterizing the stability of impulsive, sampled-data and switched systems. The goal of this paper is to show that this very same looped-functional is, in fact, ``complete'', in the sense that the resulting condition is also \emph{necessary}. This result is proved for a larger class of systems, referred to as \emph{pseudo-periodic systems}, that include periodic systems, impulsive systems, sampled-data systems and switched systems as particular cases. Along these lines, several corollaries pertaining on the analysis of these systems are directly obtained from the main result, and shown to complete those previously obtained in the literature. As the conditions are infinite-dimensional, we also describe how they can be (approximately) solved using sum of squares programming \cite{Parrilo:00}.

\textbf{Outline:} Preliminaries are stated in Section \ref{sec:prel} while the main result is presented in Section \ref{sec:main}. This result is then applied to various problems in Section \ref{sec:applications}. Finally, Section \ref{sec:comp} provides some practical discussions about the verification of infinite-dimensional semidefinite programs.

\textbf{Notations:} The sets of symmetric and positive definite matrices of dimension $n$ are denoted by $\mathbb{S}^n$ and $\mathbb{S}_{\succ0}^n$ respectively. Given two real symmetric matrices $A,B$, the expression $A\succ(\succeq) B$ means that $A-B$ is positive (semi)definite. Given a real square matrix $A$, the operator $\He(A)$ stands for the sum $A+A^{\T}$. The set of natural numbers $\{1,\ldots\}$ is denoted by $\mathbb{N}$ where the set of whole numbers is denoted by $\mathbb{N}_0:=\mathbb{N}\cup\{0\}$.

\section{Definitions}\label{sec:prel}

Some definitions and preliminary results are introduced below. In particular, pseudo-periodic functions and pseudo-periodic systems are formally characterized. Discrete-time state-transition matrices for pseudo-periodic systems are also defined.

\begin{define}[Pseudo-periodic functions]
Given scalars $0<\delta\le T_{min}\le T_{max}<\infty$, let $f:[0,T_{max}]\times[T_{min},T_{max}]\to\mathbb{R}$ be bounded and integrable, and define the family of sequence of time instants\footnote{Note that this sequence does not admit any accumulation point and is unbounded from above.}
 \begin{equation}
    \mathbb{I}:=\left\{\{t_k\}_{k\in\mathbb{N}}:\ t_{i+1}-t_i\in[T_{min},T_{max}],\ i\in\mathbb{N}_0,\ t_0=0\right\}.
  \end{equation}
  The generated class of \textbf{pseudo-periodic functions} is given by
  \begin{equation}
\begin{array}{lcl}
      \mathcal{G}_f&:=&\left\{\sum_{k=0}^{\infty}\phi_k(t):\ t\ge t_0=0,\ \{t_k\}_{k\in\mathbb{N}}\in\mathbb{I}\right\},
\end{array}
  \end{equation}
  where $\phi_k(t):=f(t-t_k,t_{k+1}-t_k)\mathds{1}_{[t_k,t_{k+1})}(t)$ and $\mathds{1}_{[t_k,t_{k+1})}$ is the indicator function of the interval $[t_k,t_{k+1})$. We also define $T_k:=t_{k+1}-t_k$ for all $k\in\mathbb{N}_0$.\menddef
\end{define}
The function $f$ above is referred to as a \emph{seed function} and the set $\mathbb{I}$ as a \emph{family of sequence of transition times}.

\begin{define}
  The linear system
  \begin{equation}\label{eq:mainsyst}
\begin{array}{lcl}
  \dot{x}(t)&=&\Sigma(t)x(t),\ t\notin\{t_k\}_{k\in\mathbb{N}}\\
  x(t)&=&Jx(t^-),\ t\in\{t_k\}_{k\in\mathbb{N}}\\
  x(0)&=&x_0
\end{array}
\end{equation}
where $x,x_0\in\mathbb{R}^n$ are the system state\footnote{The state is assumed to be right-continuous and the left-limit at $t_k$ is denoted by $x(t_k^-)=\lim_{s\uparrow0}x(t_k+s)$.} and the initial condition, is said to be an \emph{impulsive pseudo-periodic system} if the matrix $\Sigma(t)$ is a pseudo-periodic matrix function, i.e. it can be expressed as
\begin{equation}\label{eq:mainsystA}
  \Sigma(t):=\sum_{k=0}^{\infty}A(t-t_k,t_{k+1}-t_k)\mathds{1}_{[t_k,t_{k+1})}(t),\ t_0=0
\end{equation}
for some seed matrix function $A(\cdot,\cdot)$ with $\{t_k\}_{k\in\mathbb{N}}\in\mathbb{I}$.\menddef
\end{define}

It is assumed here that the matrix $\Sigma(t)$ is sufficiently regular so that solutions to the differential equation are well defined. It is immediate to see that periodic systems and impulsive systems are a particular case of pseudo-periodic systems. Note that we assume here that the impulse times and transition times coincide.

\begin{define}[State-transition matrix]
  The state-transition matrix of system (\ref{eq:mainsyst})-(\ref{eq:mainsystA}) is defined as the unique matrix function
\begin{equation}
  \Phi:[0,T_{max}]\times[T_{min},T_{max}]\to\mathbb{R}^{n\times n}
\end{equation}
that solves the parametrized family of linear differential equations
\begin{equation}\label{eq:mono2}
\begin{array}{rcl}
    \dfrac{\partial}{\partial \eta}\Phi(\eta,T)&=&A(\eta,T)\Phi(\eta,T),\ \eta\in[0,T]\\
    \Phi(0,T)&=&I_n
\end{array}
\end{equation}
for all $T\in[T_{min},T_{max}]$.\menddef
\end{define}

Using the state-transition matrix above, it is possible to define the transition map associated to the system  (\ref{eq:mainsyst})-(\ref{eq:mainsystA}). This operator plays a key role in the paper.
\begin{define}[Discrete-time transition map]
  Given a family of transition times, the discrete-time transition map corresponding to system (\ref{eq:mainsyst})-(\ref{eq:mainsystA}) is defined by $\Psi(T):=\Phi(T,T)$, $T\in[T_{min},T_{max}]$,  and we have that
  \begin{equation}\label{eq:mono1}
\begin{array}{rcl}
    x(t_{k+1}^-)&=&\Psi(T_k)Jx(t_k^-),\ k\in\mathbb{N}\\
    x(t_0)&=&x_0
\end{array}
\end{equation}
where $\{t_k\}_{k\in\mathbb{N}}\in\mathbb{I}$.
\end{define}

The discrete-time transition map hence defines the sequence of state values for any sequence of transition times in $\mathbb{I}$.

\section{Main result}\label{sec:main}

The following result is the main result of the paper that will allow us to prove that the conditions based on looped-functionals considered in \cite{Briat:12h,Briat:13b} equivalently characterize certain stability conditions taking the form of linear matrix inequalities:
\begin{theorem}\label{th:ncs}
  Let $A,J\in\mathbb{R}^{n\times n}$ and $M\in\mathbb{S}^n$ be given matrices, and let $T$ be a positive scalar. Then, the following statements are equivalent:
  \begin{enumerate}
    \item There exists a matrix $P\in\mathbb{S}_{\succ0}^n$ such that the matrix inequality
    \begin{equation}\label{eq:LMIth}
      J^{\T}\Psi(T)^{\T}P\Psi(T)J-J^{\T}PJ+M+\eps I\preceq0
    \end{equation}
    holds for some $\eps>0$ and where $\Psi(\cdot)$ is defined in \eqref{eq:mono1}.
    \item  There exist a matrix $P\in\mathbb{S}_{\succ0}^n$ and continuous matrix-valued functions $Z_1,Z_2:[0,T]\to\mathbb{S}^n$ such that the conditions
          \begin{equation}\label{eq:LMI1th}
        \begin{array}{rcl}
            Z_1(T)&=&0\\
            J^{\T}Z_{1}(0)J+Z_{2}(0)-Z_{2}(T)&=&0
  \end{array}
\end{equation}
and
\begin{equation}\label{eq:LMI2th}
  \begin{array}{rcl}
    \He[(TP+Z_1(\tau))A(\tau)]+\dot{Z}_1(\tau)&\preceq&0\\
    M+\eps I_n+\dot{Z}_2(\tau)&\preceq&0
  \end{array}
  \end{equation}
  hold  for some $\eps>0$ and for all $\tau\in[0,T]$.
  \end{enumerate}
  \end{theorem}
\begin{proof}
The proof is given in the Appendix.
\end{proof}

\begin{remark}\label{rem:1}
  The above result straightforwardly extends to the case when matrices $\Psi$, $J$ and $M$ depend on some additional time-invariant parameters, say $\rho$. In such a case, the matrix-valued functions $Z_1$ and $Z_2$ also need to depend on this additional parameter in order to preserve necessity.
\end{remark}

\section{Applications}\label{sec:applications}

To demonstrate the versatility and usefulness of the main result, we now apply it to the analysis of various systems such as pseudo-periodic systems, impulsive systems, sampled-data systems and switched systems. It is important to stress that these results have been partially obtained in \cite{Briat:12h,Briat:12h,Seuret:13}. However, only sufficiency was proven. Using Theorem \ref{th:ncs}, we demonstrate that these conditions are also necessary.

\subsection{Application to pseudo-periodic systems with impulses}

We have the following result:
\begin{theorem}\label{th:fundd2}
  The pseudo-periodic system (\ref{eq:mainsyst})-(\ref{eq:mainsystA}) with pseudo-period $T_k\in[T_{min},T_{max}]$, $0<\delta\le T_{min}\le T_{max}<\infty$, is asymptotically stable if one of the following equivalent statements hold:
  \begin{enumerate}
        \item There exists a matrix $P\in\mathbb{S}_{\succ0}^n$ such that the LMI
    \begin{equation}\label{eq:stabmono2}
        J^T\Psi(\theta)^{\T}P\Psi(\theta)J-P\prec0
    \end{equation}
  holds for all $\theta\in[T_{min},T_{max}]$.
  \item There exist a constant matrix $P\in\mathbb{S}_{\succ0}^n$, a scalar $\eps>0$ and a matrix function $Z:[0,T_{max}]\times[T_{min},T_{max}]\to\mathbb{S}^{3n}$ differentiable with respect to the first variable and verifying
\begin{equation}\label{eq:loliloolldr_ap}
Y_2^{\T}Z(\theta,\theta)Y_2-Y_1^{\T}Z(0,\theta)Y_1=0
\end{equation}
for all $\theta\in[T_{min},T_{max}]$ where
\begin{equation}\label{eq:y1y2}
\begin{array}{lclclcl}
  Y_1&=&\begin{bmatrix}
    J & 0_n\\
    I_n & 0_n\\
    0_n & I_n
  \end{bmatrix},&\quad& Y_2&=&\begin{bmatrix}
    0_n & I_n\\
    I_n & 0_n\\
    0_n & I_n
  \end{bmatrix}
\end{array}
\end{equation}
and such that the LMI
\begin{equation}\label{eq:expless_ap}
 \hspace{-5mm}\Theta(\tau,\theta)+\He\left(Z(\tau,\theta)\begin{bmatrix}
   A(\tau,\theta) & 0_n & 0_n\\
   0_n & 0_n & 0_n\\
   0_n & 0_n & 0_n
 \end{bmatrix}\right)+\dfrac{\partial}{\partial\tau}Z(\tau,\theta)\preceq0
\end{equation}
holds for all $\tau\in[0,\theta]$ and $\theta\in[T_{min},T_{max}]$ where
\begin{equation}\label{eq:dlksqdqsjdql_ap}
  \Theta(\tau,\theta):=\begin{bmatrix}
   \theta\He\left[PA(\tau,\theta)\right] & 0_n & 0_n\\
   \star & \eps I_n+J^{\T}PJ-P & 0_n\\
   \star & \star & 0_n
 \end{bmatrix}.
\end{equation}
  \end{enumerate}
\end{theorem}
\begin{proof}
From Lyapunov theory, we can see that the condition of statement (a) is a quadratic stability condition implying stability of the system. To prove the equivalence between the two statements, we consider Remark \ref{rem:1} and extends the matrix function $Z(\tau)$ to $Z(\tau,T_k)$ so that the equality condition  \eqref{eq:loliloolldr_ap} may be satisfied (see also \cite{Briat:12h,Briat:13b}). Choosing then $Z=\diag(Z_1,Z_2,0_n)$ in \eqref{eq:expless_ap} leads to a condition of the form \eqref{eq:LMI2th}, from which the equivalence follows.
%
%
%
%
\end{proof}

\subsection{Application to impulsive and sampled-data systems}

The case of impulsive systems is obtained by simply choosing the matrix $\Sigma(t)$ in \eqref{eq:mainsyst} to be a constant matrix $A\in\mathbb{R}^{n\times n}$. This leads to the following system
 \begin{equation}\label{eq:impsyst}
\begin{array}{lcl}
  \dot{x}(t)&=&Ax(t),\ t\notin\{t_k\}_{k\in\mathbb{N}}\\
  x(t)&=&Jx(t^-),\ t\in\{t_k\}_{k\in\mathbb{N}}\\
  x(t_0)&=&x_0.
\end{array}
\end{equation}
and the following result completing the one derived in \cite{Briat:12h}:
\begin{corollary}\label{cor:imp}
  The impulsive system \eqref{eq:impsyst} with inter-impulses times $T_k:=t_{k+1}-t_k\in[T_{min},T_{max}]$, $0<\delta\le T_{min}\le T_{max}<\infty$, is asymptotically stable if one of the following equivalent statements hold:
  \begin{enumerate}
        \item There exists a matrix $P\in\mathbb{S}_{\succ0}^n$ such that the LMI
    \begin{equation}
        J^{\T}e^{A^{\T}\theta}Pe^{A\theta}J-P\prec0
    \end{equation}
  holds for all $\theta\in[T_{min},T_{max}]$.
  \item There exist a constant matrix $P\in\mathbb{S}_{\succ0}^n$, a scalar $\eps>0$ and a matrix function $Z:[0,T_{max}]\times[T_{min},T_{max}]\to\mathbb{S}^{3n}$ differentiable with respect to the first variable and verifying
\begin{equation}
Y_2^{\T}Z(\theta,\theta)Y_2-Y_1^{\T}Z(0,\theta)Y_1=0
\end{equation}
for all $\theta\in[T_{min},T_{max}]$ where $Y_1$ and $Y_2$ are defined in \eqref{eq:y1y2}
and such that the LMI
\begin{equation}
 \Theta(\theta)+\He\left(Z(\tau,\theta)\begin{bmatrix}
   A & 0_n & 0_n\\
   0_n & 0_n & 0_n\\
   0_n & 0_n & 0_n
 \end{bmatrix}\right)+\dfrac{\partial}{\partial\tau}Z(\tau,\theta)\preceq0
\end{equation}
hold for all $\tau\in[0,\theta]$ and $\theta\in[T_{min},T_{max}]$ where
\begin{equation}
  \Theta(\theta):=\begin{bmatrix}
   \theta\He\left[PA\right] & 0_n & 0_n\\
   \star & \eps I_n+J^{\T}PJ-P & 0_n\\
   \star & \star & 0_n
 \end{bmatrix}.
\end{equation}
  \end{enumerate}
\end{corollary}
For numerical results obtained using the above corollary, see \cite{Briat:12h}.

\begin{remark}
It is known that sampled-data systems can be reformulated as impulsive systems; see e.g. \cite{Goebel:09}. In this respect, the result above can be applied to sampled-data systems as well. See \cite{Seuret:13} for some numerical results about sampled-data systems using looped-functionals.
\end{remark}

\subsection{Application to switched systems}

Let us consider now switched systems of the form
\begin{equation}\label{eq:swsyst}
\begin{array}{rcl}
    \dot{x}(t)&=&A_{\sigma(t)}x(t)\\
    x(0)&=&x_0
\end{array}
\end{equation}
where $\sigma:\mathbb{R}_{\ge0}\to\{1,\ldots,N\}$ and assume that the switching times are given by the sequence $\{t_k\}_{k\in\mathbb{N}}$. We then have the following result completing the one obtained in \cite{Briat:13b}:

\begin{corollary}
  The switched system \eqref{eq:swsyst} is asymptotically stable for any dwell-time $t_{k+1}-t_k=:T_k\in[T_{min},T_{max}]$,  $0<\delta\le T_{min}\le T_{max}<\infty$, $k\in\mathbb{N}$, if one of the following equivalent statements hold:
  \begin{enumerate}
        \item There exist matrices $P_i\in\mathbb{S}_{\succ0}^n$, $i=1,\ldots,N$ such that the LMIs
    \begin{equation}
        e^{A_i^{\T}\theta}P_ie^{A_i\theta}-P_j\prec0
    \end{equation}
  holds for all $\theta\in[T_{min},T_{max}]$ and for all $i,j=1,\ldots,N$, $i\ne j$.
  \item There exist constant matrices $P_i\in\mathbb{S}_{\succ0}^n$, $i=1,\ldots,N$, a scalar $\eps>0$ and matrix functions $Z_{ij}:[0,T_{max}]\times[T_{min},T_{max}]\to\mathbb{S}^{3n}$, $i,j=1,\ldots,N$, $i\ne j$, differentiable with respect to the first variable and verifying
\begin{equation}
W_2^{\T}Z_{ij}(\theta,\theta)W_2-W_1^{\T}Z_{ij}(0,\theta)W_1=0
\end{equation}
for all $\theta\in[T_{min},T_{max}]$ where
\begin{equation}
\begin{array}{lclclcl}
  W_1&=&\begin{bmatrix}
    I_n & 0_n\\
    I_n & 0_n\\
    0_n & I_n
  \end{bmatrix},&\quad& W_2&=&\begin{bmatrix}
    0_n & I_n\\
    I_n & 0_n\\
    0_n & I_n
  \end{bmatrix}
\end{array}
\end{equation}
and such that the LMI
\begin{equation}
 \hspace{-4mm}\Theta_{ij}(\theta)+\He\left(Z_{ij}(\tau,\theta)\begin{bmatrix}
   A_i & 0_n & 0_n\\
   0_n & 0_n & 0_n\\
   0_n & 0_n & 0_n
 \end{bmatrix}\right)+Z_{ij}^\prime(\tau,\theta)\preceq0
\end{equation}
hold for all $\tau\in[0,\theta]$, all $\theta\in[T_{min},T_{max}]$ and for all $i,j=1,\ldots,N$, $i\ne j$, where
\begin{equation}
  \Theta_{ij}(\theta):=\begin{bmatrix}
   \theta\He\left[PA_i\right] & 0_n & 0_n\\
   \star & \eps I_n+P_i-P_j & 0_n\\
   \star & \star & 0_n
 \end{bmatrix}.
\end{equation}
  \end{enumerate}
\end{corollary}
\begin{proof}
  Based on Theorem \ref{th:fundd2}, we can define $A(\tau,T)=A_i$, $Z=Z_{ij}$ and $M=M_{ij}:=P_i-P_j$ for all $i,j=1,\ldots,N$, $i\ne j$. The result then follows.
\end{proof}

For numerical results obtained using the above corollary, see \cite{Briat:13b}.

\section{Computational considerations}\label{sec:comp}

The conditions stated in the above results are infinite-dimensional semidefinite programs that cannot be checked directly. A way for turning these conditions into finite-dimensional ones is to use sum of squares techniques \cite{Parrilo:00,sostools3}. To use this approach, we need to assume that all the infinite-dimensional variables are polynomials. Note that the choice of using polynomials can be justified by the fact that polynomials are dense in the set of continuous functions with compact support. Once the conditions have been reformulated in the sum of squares paradigm, then the problem can be turned into an equivalent finite-dimensional LMI problem, using for instance SOSTOOLS \cite{sostools3}, which can be solved in turn using standard SDP solvers.

In what follows, we say that a polynomial symmetric matrix $S(y)$, $y\in\mathbb{R}^n$, is a sum of squares matrix if there exists a (possibly very tall) polynomial matrix $T(y)$ such that $S(y)=T(y)^{\T}T(y)$. Below is the sum of squares formulation for the conditions of statement (b) of Corollary \ref{cor:imp}:
\begin{proposition}
  Assume that there exist a matrix $P\in\mathbb{S}^n_{\succ0}$ and polynomial matrices $Z,\Gamma_1,\Gamma_2:\mathbb{R}^2\to\mathbb{S}^{3n}$ such that
\begin{itemize}
  \item $\Gamma_1$ and $\Gamma_2$ are sum of squares matrices.
  \item The equality $Y_2^{T}Z(\theta,\theta)Y_2-Y_1^{\T}Z(0,\theta)Y_1=0$ holds for all $\theta\in\mathbb{R}$.
  \item The expression
\begin{equation}
\begin{array}{l}
   -\Theta(\theta)-\He\left(Z(\tau,\theta)\begin{bmatrix}
   A & 0_n & 0_n\\
   0_n & 0_n & 0_n\\
   0_n & 0_n & 0_n
 \end{bmatrix}\right)-\dfrac{\partial}{\partial\tau}Z(\tau,\theta)\\
    -\Gamma_1(\tau,\theta)\tau(\theta-\tau)-\Gamma_2(\tau,\theta)(T_{max}-\theta)(\theta-T_{min})
\end{array}
\end{equation}
is a sum of squares matrix where $\Theta(\theta)$ is defined in Corollary \ref{cor:imp}.
\end{itemize}
Then, the impulsive system \eqref{eq:impsyst} with inter-impulses times $T_k:=t_{k+1}-t_k\in[T_{min},T_{max}]$, $0<\delta\le T_{min}\le T_{max}<\infty$, is asymptotically stable.
\end{proposition}

The role of the SOS matrices $\Gamma_1$ and $\Gamma_2$ is to incorporate the constraints that $\tau\in[0,\theta]$ and $\theta\in[T_{min},T_{max}]$ inside the conditions.

\section{Conclusion}

A proof for the necessity of previously obtained looped-functional conditions have been proposed. It is shown that the conditions obtained for the analysis of switched, sampled-data and impulsive systems are necessary and sufficient for the characterization of the associated discrete-time stability conditions. This result therefore consolidates the framework of looped-functionals.

\appendix

\renewcommand\thesection{\Alph{section}}

\section{Proof of Theorem \ref{th:ncs}}\label{ap:proof}

\subsection{A preliminary result}

\begin{lemma}[\cite{Gajic:95}]\label{lem:sol}
  Let $U:\mathbb{R}_{\ge0}\to\mathbb{R}^{n\times n}$, $Q:\mathbb{R}_{\ge0}\to\mathbb{S}^n$ and $S_0\in\mathbb{S}^n$. The symmetric solution $S:\mathbb{R}_{\ge0}\to\mathbb{S}^n$ of the matrix equality
  \begin{equation}
  \begin{array}{rcl}
    -\dot{S}(t)+U(t)^{\T}S(t)+S(t)U(t)&=&Q(t),\ t\ge0\\
    S(0)&=&S_0
  \end{array}
  \end{equation}
  is given by
  \begin{equation}
    S(t)=X_U(t,0)^{\T}S_0X_U(t,0)+\int_0^tX_U(t,s)^{\T}Q(s)X_U(t,s)ds
  \end{equation}
  where $X_U(t,s)=Y_U(t)Y_U(s)^{-1}$ where $\dot{Y}_U(t)=U(t)Y_U(t)$, $Y_U(0)=I$. When $U(t)\equiv U$ is a constant matrix, then $X_U(t,s)=e^{U(t-s)}$. 
\end{lemma}

\subsection{Proof of Theorem \ref{th:ncs}}

Instead of proving the equivalence between the conditions of statement (b) and statement (a), we first make a change of variable to turn the conditions of statement (a) into a more convenient form. To this aim, let us then define the change of variables $\tilde{Z}_i(\tau)=Z_i(T-\tau)$. Then, clearly, we have that $\tilde{Z}_i(0)=Z_i(T)$ and $\tilde{Z}_i(T)=Z_i(0)$. Therefore, we obtain that the conditions of statement (b) are equivalent to the conditions
 \begin{equation}\label{eq:change1}
        \begin{array}{rcl}
            \tilde{Z}_1(0)&=&0\\
            J^{\T}\tilde{Z}_{1}(T)J+\tilde{Z}_{2}(T)-\tilde{Z}_{2}(0)&=&0
  \end{array}
\end{equation}
and
\begin{equation}\label{eq:change2}
  \begin{array}{rcl}
    \He[(TP+\tilde{Z}_1(\tau))A(\tau)]-\dot{\tilde{Z}}_1(\tau)&\preceq&0\\
    M+\eps I_n-\dot{\tilde{Z}}_2(\tau)&\preceq&0
  \end{array}
  \end{equation}
  hold for all $\tau\in[0,T]$.

\textbf{Proof of (b) $\Rightarrow$ (a).} Assume that the conditions \eqref{eq:change1} and \eqref{eq:change2} hold. Then, using Lemma \ref{lem:sol}, we can show that the inequalities \eqref{eq:change2} imply that
\begin{equation}
  \tilde{Z}_2(\tau)\succeq \tilde{Z}_2(0)+\tau(M+\eps I_n)
\end{equation}
together with
\begin{equation}
  \begin{array}{lcl}
    \tilde{Z}_1(\tau)&\succeq&\Psi(\tau)^{\T}\tilde{Z}_1(0)\Psi(\tau)\\
    &&+T\int_0^TX_A(\tau,s)^{\T}\He[PA(s)]X_A(\tau,s)ds\\
    &=&\Psi(\tau)^{\T}\tilde{Z}_1(0)\Psi(\tau)\\
    &&-T\int_0^T\dfrac{\partial}{\partial s}[X_A(\tau,s)^{\T}PX_A(\tau,s)]ds\\
    &=&\Psi(\tau)^{\T}\tilde{Z}_1(0)\Psi(\tau)+T\left[\Psi(\tau)^{\T}P\Psi(\tau)-P\right],
  \end{array}
\end{equation}
where we used the fact that $\Psi(\tau)=X_A(\tau,0)$. The first equality in \eqref{eq:change1} implies that we have
\begin{equation}
  \tilde{Z}_1(\tau)\succeq T\left[\Psi(\tau)^{\T}P\Psi(\tau)-P\right].
\end{equation}
Finally, the second equality implies
\begin{equation}
\begin{array}{l}
    T\left(J^{\T}\left[\Psi(T)^{\T}P\Psi(T)-P\right]J+M+\eps I_n\right)\\
    \qquad\qquad\qquad\preceq J^{\T}\tilde{Z}_{1}(T)J+\tilde{Z}_{2}(T)-\tilde{Z}_{2}(0)=0
\end{array}
\end{equation}
where the equality to 0 holds by assumption. This then implies that the condition \eqref{eq:LMIth} holds, proving then sufficiency.

\textbf{Proof of (a) $\Rightarrow$ (b).} The necessity can be proven by explicitly constructing suitable matrix-valued $Z_1$ and $Z_2$ that verify the conditions of statement (b) whenever the condition of statement (a) holds. Let us therefore consider the following matrix-valued functions
\begin{equation}
  \begin{array}{lcl}
    \tilde{Z}_1^*(\tau)&=&T\left(\Psi(\tau)^{\T}P\Psi(\tau)-P\right)\\
    \tilde{Z}_2^*(\tau)&=&\tilde{Z}_2^*(0)+(M+\eps I_n)\tau-\tau\Xi,
  \end{array}
\end{equation}
where $\Xi:=J^{\T}\Psi(T)^{\T}P\Psi(T)J-TJ^{\T}PJ+M+\eps I_n\preceq0$. Note that the inequality holds by assumption. Substituting these expressions in \eqref{eq:change2} yields
\begin{equation*}
    \He[(TP+\tilde{Z}^*_1(\tau))A(\tau)]-\dot{\tilde{Z}}^*_1(\tau)=0
\end{equation*}
for the first condition, and
\begin{equation*}
   \begin{array}{rcl}
    M+\eps I_n-\dot{\tilde{Z}}^*_2(\tau)&=&\Xi\preceq0
  \end{array}
\end{equation*}
for the second condition. We can therefore clearly see that the proposed functions verify the conditions \eqref{eq:change2} provided that $\Xi\preceq0$ holds; i.e. the condition of statement (a) holds.

Let us now consider the conditions \eqref{eq:change1}. It is immediate to see that the first condition in \eqref{eq:change1} is trivially satisfied. The second one given by
\begin{equation*}
\begin{array}{lcl}
  J^{\T}\tilde{Z}^*_{1}(T)J+\tilde{Z}^*_{2}(T)-\tilde{Z}^*_{2}(0)&=&T\Xi-T\Xi\\
  &=&0
\end{array}
\end{equation*}
is also satisfied. Necessity is proved.

\bibliographystyle{IEEEtran}

\end{document}